\documentclass[conference]{IEEEtran}

\usepackage{graphicx}
\usepackage[all]{xy}
\usepackage{amsmath}
\usepackage{amssymb}
\usepackage{amsthm}
\usepackage{proof}
\usepackage{cite}
\usepackage{hyperref}
\usepackage{url}

\newtheorem{theorem}{Theorem}
\newtheorem{lemma}[theorem]{Lemma}

\newtheorem{proposition}[theorem]{Proposition}
\newtheorem{definition}[theorem]{Definition}
\newtheorem{example}[theorem]{Example}
\newtheorem{remark}[theorem]{Remark}

\newcommand{\base}[1]{\mathcal{#1}}
\newcommand{\supp}{\Vdash}
\newcommand{\baseext}{\supseteq}
\newcommand{\xacml}{XACML}
\newcommand{\asp}{ASP}
\newcommand{\enc}[1]{\ulcorner #1 \urcorner}

\begin{document}

\title{Verification of Robust Properties \\ for Access Control Policies}

\author{
  \IEEEauthorblockN{Alexander V. Gheorghiu}
  \IEEEauthorblockA{
    School of Electronics and Computer Science\\
    University of Southampton, Southampton, UK\\
    \textit{and}\\
    Department of Computer Science\\
    University College London, London, UK\\
    Email: a.v.gheorghiu@soton.ac.uk
  }
}

\maketitle

\begin{abstract}
Existing methods for verifying access control policies require the policy
to be complete and fully determined before verification can proceed,
but in practice policies are developed iteratively, composed from
independently maintained components, and extended as organisational
structures evolve. We introduce \emph{robust property verification} ---
the problem of determining what a policy's structure commits
it to regardless of how pending decisions are resolved and regardless of
subsequent extension --- and develop a semantic framework based on
\emph{base-extension semantics} to address it. We define a support
judgment $\Vdash_{\base{P}} \phi$, stating that policy $\base{P}$ has robust property $\phi$, with connectives for implication,
conjunction, disjunction, and negation, prove that it is
\emph{compositional} (verified properties persist under policy extension
by a monotonicity theorem), and show that despite quantifying universally
over all possible policy extensions the judgment reduces to proof-search
in a second-order logic programming language. Soundness and completeness of
this reduction are established, yielding a finitary and executable
verification procedure for robust security properties.
\end{abstract}


% ============================================================================
\section{Introduction}
\label{sec:intro}

Access control is one of the most fundamental mechanisms in computer
security. It governs which subjects may perform which actions on which
resources: whether a user may read a file, whether a process may invoke
a service, whether a principal may delegate authority to another. Every
serious computing system depends on access control, and the consequences
of getting it wrong range from data exposure to privilege escalation to
the complete undermining of a system's security guarantees. As systems
have grown in scale and complexity, the policies that implement access
control have grown correspondingly: from simple matrices associating
users with permissions, to rich, structured languages such as
\xacml{}~\cite{xacml-standard} that support hierarchical policy sets,
combining algorithms, attribute-based conditions, and dynamic
delegation. This expressiveness is necessary. It is also a source of
error.

Automated verification of access control policies has therefore become
an active concern. The central question is whether a given policy
satisfies a given security property: that no subject can both author and
review the same paper, that a deny decision from any component overrides
all permits, that delegation cannot circumvent a conflict-of-interest
constraint. Two broad approaches have emerged. Model-theoretic
methods---using Alloy~\cite{jackson-alloy}, binary decision
diagrams~\cite{verification-change-impact}, or SMT
solvers~\cite{bouchet2020}---translate policies into a mathematical
model and check whether the property holds in that model. 

Logic-based
methods represent policies as logical theories and properties as
formulae, reducing verification to theorem-proving or
satisfiability~\cite{ramli-logic-xacml,ahn-asp-xacml}. Both traditions
have produced capable tools: Margrave checks change-impact for \xacml{}
policies using decision diagrams; Zelkova verifies trust-safety
properties of AWS S3 policies using SMT; Answer Set Programming encodes
\xacml{} under stable model semantics~\cite{gelfond-lifschitz-stable}
and queries the resulting program. These tools address the verification
problem as it has traditionally been posed, and they address it well.

The traditional formulation, however, rests on an assumption that is
rarely made explicit: that the policy is \emph{complete}. That is, the policy must be fully specified before
verification can begin. A property is
verified by checking whether it holds in the model that the policy
defines, or by checking whether it is a consequence of the theory the
policy determines. In the model-theoretic tradition, a pending
decision---an unresolved role assignment, a delegation not yet
settled---has no natural representation: models are complete structures,
and an incomplete policy does not determine one. In the logical
tradition, the situation is similar: \asp{} verification proceeds by
computing stable models, and each stable model commits to specific
resolutions of any alternatives present in the rules. Reasoning about
what holds across all resolutions still requires enumerating them.

In practice, policies are developed iteratively, composed from components
maintained by different parties, and extended as organisational
structures evolve. A security administrator may need to enforce a
constraint before the relevant
assignments have been finalised. The policy's rules may already
determine that the constraint will hold, whatever resolution is
eventually chosen, but existing tools cannot say so. They can only check
particular completed instances, one at a time. When pending decisions
multiply, this becomes intractable; and when the policy is later
extended, the verification work must be redone from scratch.

What is missing, then, is a framework in which verification addresses
not a fixed completed policy but the policy's commitments
under hypothetical extension. The right questions are not only ``does
this property hold in the current policy?'' but ``does the policy's
structure guarantee this property no matter how it is developed?''. These are questions about the behaviour of a
policy under hypothetical change---about what follows not from any
particular completion, but from the policy's rules themselves. Answering
them requires new semantic and formal methods: a semantics that
interprets properties not in static models but across extensions of a
policy, and a logic whose connectives express commitment
rather than truth in a completed state.

Three requirements constrain what such a framework must look like. First, it should be \emph{inferential}: security properties should be grounded in the policy's derivation structure rather than in models of a completed state, so that pending decisions are not treated as gaps to be filled but as constraints to be propagated. Second, it should be \emph{compositional}: a property verified against a policy should remain verified when the policy is extended with new rules, so that administrators are not forced to re-verify from scratch each time an assignment is finalised or a component is updated. Third, it should be \emph{finitary}: although robustness is defined by quantifying over all possible extensions of a policy ---an infinite family --- verification should reduce to a decidable procedure, so that the framework is usable in practice rather than merely in principle. 

This paper develops such a framework using \emph{proof-theoretic semantics}~\cite{schroeder-heister-pts}. This raises the problem of \emph{robust
properties}: security constraints guaranteed to hold regardless of how
pending decisions are resolved. To define it formally, we use a support
$\supp_\base{P} \phi$ that says $\phi$ is a robust property of policy $\base{P}$. The judgment is \emph{compositional}:
verified robust properties persist under policy extension without
re-verification, by a monotonicity theorem. We show that robust property
verification reduces to proof-search in a second-order logic programming
language, and prove soundness and completeness. Our framework targets an
earlier stage of the policy lifecycle than tools such as Zelkova and
Margrave, where decisions remain pending and compositional guarantees
are needed before the policy is finalised.

\subsection*{Contributions}

\begin{itemize}
\item We identify \emph{robust property verification} for access control policies
  as a proof-theoretic problem.
\item We define the support judgment $\Vdash_{\base{P}} \phi$ for fulfilling the three desiderata of inferentiality, compositionality, and
  finitariness.
  
\item We show that robust property verification reduces to proof-search in a
  second-order logic programming language and prove \emph{soundness and
  completeness} (Theorems~\ref{thm:soundness} and~\ref{thm:completeness}).
\end{itemize}

\subsection*{Organisation}

Section~\ref{sec:background} situates the work in the context of existing approaches to access control verification and introduces the conference management system that serves as our running example throughout. Section~\ref{sec:robust} develops the semantic framework: we define the support judgment $\supp_\base{P} \phi$, introduce its connectives for implication, disjunction, conjunction, and negation, and establish the monotonicity theorem that grounds compositional verification. Section~\ref{sec:operational} turns to computation: we introduce the Horn clause policy language, define the goal encoding $\enc{\cdot}$, and prove the soundness and completeness theorem reducing robust property verification to proof-search. Section~\ref{sec:conclusion} discusses related and future work.

\section{Background}
\label{sec:background}

Access control policies specify which subjects may perform which actions
on which resources. To verify properties of such policies, one has to represent them and the properties. To this end, symbolic logic has been a natural tool for their 
analysis~\cite{abadi-calculus,abadi-survey}.

\subsection{Logic and Access Control Policies}

Whatever their surface syntax, policies share a
common structure: they consist of rules that derive authorisation
decisions from facts about subjects, roles, resources, and their
relationships.  This inferential structure is naturally captured by a
logic program --- typically, a set of Horn clauses whose proof-search corresponds
to the policy's decision procedure.

Some policy formalisms make this logical structure explicit from the
outset. The modal logic of principals due to Abadi, Burrows, Lampson,
and Plotkin~\cite{abadi-calculus} treats authorisation as derivation
directly: a request is granted if and only if it follows by a formal
proof from the policy's axioms about trust and delegation. Role-based
access control models have been given first-order and deontic
formalisations in the same spirit~\cite{samarati-vimercati,nist2016}.
Abadi's survey~\cite{abadi-survey} traces several such approaches and
their shared commitment to the paradigm that authorisation is proof.

Other policy languages require translation. \xacml{}~\cite{xacml-standard}
is an industrial standard: an XML-based language in which a
policy decision point evaluates hierarchically structured
\emph{PolicySets}, \emph{Policies}, and \emph{Rules} against access
requests to produce a decision of \texttt{Permit}, \texttt{Deny},
\texttt{NotApplicable}, or \texttt{Indeterminate}, with multiple
applicable rules reconciled by \emph{combining algorithms} such as
deny-overrides. Despite its XML syntax, \xacml{} has a well-understood
logical content: Ramli et al.~\cite{ramli-logic-xacml} derive a
many-sorted first-order logic that precisely captures \xacml{} 3.0, and
Ahn et al.~\cite{ahn-asp-xacml} show that \xacml{} policies can be
translated into Answer Set Programs under stable model
semantics~\cite{gelfond-lifschitz-stable}. Either way, the policy
becomes a logic program and verification becomes proof-search.

This paper takes that logical representation as its starting point. We
assume that the policy's inferential structure has been expressed as a
logic program --- whether it originated as one or was translated into one
--- and ask a question that existing approaches cannot answer: what
properties does the policy's inferential structure \emph{guarantee},
regardless of how pending decisions are resolved? The rest of the paper
develops the framework for answering this question.

\subsection{Verification of Access Control Policies}
\label{sec:related-work}

The verification of properties for access control policies has been
approached from various semantic and computational perspectives. We take
\xacml{} as our running example and reference point throughout.

On the model-theoretic side, Hughes and Bultan~\cite{hughes-bultan-2004}
translate \xacml{} into Alloy for SAT-based bounded model checking. Fisler
et al.~\cite{verification-change-impact} developed Margrave using
Multi-Terminal Binary Decision Diagrams, supporting change-impact analysis
between policy versions. Bouchet et al.~\cite{bouchet2020} deployed
Zelkova at AWS using SMT to verify trust-safety properties at scale. All
of these verify properties by checking whether they hold in a complete
model of the policy; they cannot express that a property holds across all
resolutions of a pending decision without enumerating each resolution
explicitly, and policy extension invalidates the model.

On the proof-theoretic side, Ahn et al.~\cite{ahn-asp-xacml} and Ramli
et al.~\cite{ramli-logic-xacml} work directly with the logical
representation of \xacml{} policies, reducing verification to ASP solving
or first-order theorem-proving. These approaches are closer in spirit to
ours, but they are still non-compositional: verification computes models
of a fully-determined program, every extension demands a fresh
computation, and reasoning about pending decisions requires enumerating
the alternatives case by case.

Both families of approach are effective when the policy is \emph{complete} and
fully specified. What they cannot do is reason about what a policy's
inferential structure \emph{commits} to before all decisions have been
resolved --- the problem this paper addresses. We take the logic program
representation of a policy as given and ask what it structurally
guarantees across all possible extensions. The support judgment $\supp_P
\phi$ answers this question, and its monotonicity
(Proposition~\ref{prop:monotonicity}) makes the framework compositional:
verified properties persist under extension without re-verification. Our
approach targets the stage of the policy lifecycle before decisions are
finalised, and is complementary to tools such as Zelkova and Margrave
that operate on fully-specified deployed policies.

\subsection{Running Example: Conference Management}
\label{sec:conference-example}

We use a conference paper review system as our running example. This
benchmark has been studied extensively in the verification
literature~\cite{zhang-synthesis,guelev-model-checking,verification-change-impact}
because it exhibits the full range of challenges that arise in realistic
policies: role-based access control, conflict-of-interest constraints,
delegation, and dynamic permission assignment.

\begin{example}[Conference Management System]
\label{ex:conference}
The system involves \emph{Authors} (who submit papers), \emph{PC Members}
(who review papers), a \emph{Chair} (who manages the review process), and
\emph{Sub-reviewers} (to whom PC members may delegate reviews). Key
constraints to verify include:
\begin{itemize}
\item \emph{Conflict-of-interest}: no one can review a paper they
  co-authored;
\item \emph{Separation of duty}: the Chair should not also be an author;
\item \emph{Review confidentiality}: reviewers cannot see others' reviews
  before submitting their own;
\item \emph{Delegation safety}: sub-reviewers must not have conflicts of
  interest.
\end{itemize}
As Fisler et al.\ note, a real implementation of such a system revealed
``role-conflict errors during testing''~\cite{verification-change-impact},
demonstrating the practical importance of formal verification.
\end{example}

% ============================================================================
\section{Specification of Robust Properties}
\label{sec:robust}
% ============================================================================

We now turn to the central problem of the paper: robust access control properties. The heuristic idea is simple --- a 
property is robust if it holds no matter how the policy is completed 
or extended. The purpose of this section is to make that precise in 
a way that supports formal verification.

To this end, we define a judgment 
\[
\supp_\base{P} \phi
\]
that says that the robust property $\phi$ holds of the policy $\base{P}$. We introduce connectives for implication, disjunction and conjunction, showing how
each captures a distinct and independently important aspect of robust
policy reasoning. Finally, we define negation as a limiting case and
collect the full set of clauses. 

\subsection{Support, Derivability, and Policy Extension}
\label{sec:support-basic}

We work over a set of ground atoms $\mathbb{A}$ representing
access-relevant statements: facts such as $\mathit{pcMember}(\mathit{alice})$,
$\mathit{author}(\mathit{alice}, p_{007})$, and
$\mathit{conflicted}(\mathit{alice}, p_{007})$ in the conference
management system of Example~\ref{ex:conference}. We write
$\vdash_{\base{P}} A$ to mean that the atom $A$ is derivable from
policy $\base{P}$ by application of its rules: for instance,
$\vdash_{\base{C}} \mathit{conflicted}(\mathit{alice}, p_{007})$ holds
in the program $C$ of Example~\ref{ex:conference} because
Alice is a PC member who authored $p_{007}$, and the conflict rule
fires. 

This derivability relation is the ground floor of the
semantics: it captures what the policy can establish right now,
without any hypothetical extension. We lift it to the support judgment
for atomic formulas by setting
\[
  \supp_\base{P} A \quad\text{iff}\quad \vdash_\base{P} A
  \qquad\text{(for atomic } A\text{).}
\]
The challenge is to extend support to compound security properties ---
conditions that are not themselves derivable from the policy in its
current state but that the policy's rules nonetheless \emph{commit} it
to. That commitment is expressed through the policy's behaviour under
hypothetical extension.

\paragraph{Policy extension.}
We write $\base{Q} \baseext \base{P}$ to say that policy $\base{Q}$
\emph{extends} policy $\base{P}$. That is, the rules of $\base{P}$ are a subset of the
rules of $\base{Q}$, and $\base{Q}$ may add arbitrarily many new rules.
Extension is monotone: anything derivable in $\base{P}$ remains
derivable in any $\base{Q} \baseext \base{P}$, reflecting the
assumption that policy rules, once established, are not retracted. The
set of all policies extending $\base{P}$ represents the space of all
possible ways the policy might be completed or elaborated in the future.

The conditional support judgment $\Gamma \supp_\base{P} \phi$ says that
$\phi$ is a robust consequence of $\Gamma$ over $\base{P}$: in any
extension of the policy that makes every formula in $\Gamma$ hold,
$\phi$ holds as well. Formally:
\[
\begin{array}{lcl}
  \Gamma \supp_\base{P} \phi
   &\text{iff} &
  \text{for any } \base{Q} \baseext \base{P},\;\\ & &
  \text{if } \supp_\base{Q} \psi \;\text{for all } \psi \in \Gamma,
  \text{ then } \supp_\base{Q} \phi.
\end{array}
\]
The quantification is over \emph{all} extensions --- we are reasoning
about inferential commitments across the entire future policy lifecycle,
not about any particular completion.

\paragraph{Implication.}
 We internalise conditional support as a robust implication, writing
 \[ 
 \supp_\base{P} \phi \Rightarrow \psi 
 \qquad \mbox{iff} \qquad \phi \supp_\base{P} \psi
 \]
That is,  the policy supports $\phi \Rightarrow \psi$ when every extension that
supports $\phi$ also supports $\psi$. This is not classical material
implication, which would hold vacuously whenever $\phi$ is currently
unprovable. It is an inferential conditional: the policy's rules
\emph{commit} it to $\psi$ whenever $\phi$ is established, regardless
of when or how that establishment occurs.

\begin{example}[Delegation implies authority check]
\label{ex:implication}
Let $\base{P}$ contain the rule that a PC member can delegate a paper
to a sub-reviewer only if the sub-reviewer is not conflicted. The
robust implication
\[
  \supp_\base{P}\;
  \mathsf{delegate}(m, r, p)
  \Rightarrow \neg\mathsf{conflicted}(r, p)
\]
holds because any extension of $\base{P}$ that derives a delegation
fact must, by the delegation rule, also derive the absence of conflict.
This property is not merely currently true; it is structurally
guaranteed to persist through any future extension. The hypothesis
$\mathsf{delegate}(m,r,p)$ need not be in the policy's current
database for the guarantee to hold.
\end{example}

\subsection{Disjunction: Reasoning Under Pending Alternatives}
\label{sec:disjunction}

A recurring situation in policy development is that administrators must
enforce security constraints before all role assignments have been
finalised. The conference organizers know that one of
$\{\mathsf{Alice}, \mathsf{Bob}, \mathsf{Carol}\}$ will serve as
program chair, but the appointment is pending. The separation-of-duty
constraint --- that the chair must not be an author --- must be
verified before the appointment is made. Existing approaches cannot express this internally. 

What is needed is a way to express that all security-relevant
consequences that would follow from each alternative individually
already follow from the policy itself, without committing to any
particular resolution. We call this \emph{robust disjunction}:
\[
\begin{array}{lcl}
\supp_\base{P}\, \phi \oplus \psi & \mbox{iff} & \text{for any 
$\base{Q} \baseext \base{P}$ and any atomic $A$,}\\ & &\text{if $\phi
\supp_\base{Q} A$ and $\psi \supp_\base{Q} A$, then $\vdash_\base{Q}
A$.}
\end{array}
\]
The intuition is that $\phi \oplus \psi$ acts as a
\emph{case-analysis eliminator}: the policy commits to a consequence
$A$ when that consequence is forced by each alternative, in every
possible extension. Crucially, $\supp_\base{P}\, \phi \oplus \psi$
does \emph{not} assert that one of $\phi$ or $\psi$ is currently
derivable --- neither may be. What it asserts is an inferential
commitment: the policy is positioned so that when either alternative is
resolved, the same consequences flow.

\emph{Why not classical or intuitionistic disjunction?}
Classical disjunction $\phi \vee_c \psi$ requires that at least
one disjunct is true in a completed model. But the assignment is
pending and no such model exists. \emph{Intuitionistic disjunction}
$\phi \vee_i \psi$ requires a proof that one specific disjunct holds
--- equally inapplicable. Both ask us to commit to a particular
alternative, which is precisely what we cannot yet do. Robust
disjunction asks only that the inferential content that would follow
from each alternative is already guaranteed by the policy's structure.

\begin{example}[Separation of duty under pending appointment]
\label{ex:disjunction}
Suppose $\base{P}$ contains rules specifying that no one may review a
paper they co-authored, and that the chair makes review assignments.
Alice, Bob, and Carol have each submitted papers. To verify the
separation-of-duty constraint robustly, we check that in any extension
$\base{Q} \baseext \base{P}$:
\begin{align*}
  \mathsf{chair}(\mathsf{alice})
  &\supp_\base{Q} \neg\mathsf{canReview}(\mathsf{alice}, p_{007}), \\
  \mathsf{chair}(\mathsf{bob})
  &\supp_\base{Q} \neg\mathsf{canReview}(\mathsf{bob},   p_{007}), \\
  \mathsf{chair}(\mathsf{carol})
  &\supp_\base{Q} \neg\mathsf{canReview}(\mathsf{carol}, p_{007}).
\end{align*}
If all three hold in every such extension, then by the semantics of
$\oplus$ we obtain $\supp_\base{P}\, \mathsf{chair}(\mathsf{alice})
\oplus \mathsf{chair}(\mathsf{bob}) \oplus \mathsf{chair}(\mathsf{carol})$
and consequently $\supp_\base{P}\,
\neg\mathsf{canReview}(\mathsf{chair}, p_{007})$. The constraint is
verified \emph{once}, against the policy's rules, without constructing
or examining any completed assignment.

Equally important is what happens when the pool of candidates later
expands to include Dave. By Proposition~\ref{prop:monotonicity}, the
already-verified judgments for Alice, Bob, and Carol remain valid. It
suffices to check the single additional judgment
$\mathsf{chair}(\mathsf{dave}) \supp_\base{Q}
\neg\mathsf{canReview}(\mathsf{dave}, p_{007})$. The previously
established guarantee and the work behind it need not be revisited.
\end{example}

\subsection{Conjunction: Compound Security Requirements}
\label{sec:conjunction}

Whereas disjunction models uncertainty over which alternative will
obtain, conjunction models the requirement that multiple security
properties hold simultaneously --- and that the policy is committed to
their \emph{joint} inferential consequences, not merely to each in
isolation.

Most meaningful access control properties are compound. The conference
policy must simultaneously enforce conflict-of-interest avoidance,
review confidentiality, and delegation safety. Verifying each property
in isolation can miss failures that only manifest when they are
considered together: an override rule may be consistent with each
individual requirement yet break their combination. We call this \emph{robust conjunction:}
\[
\begin{array}{lcl}
\supp_\base{P}\, \phi \otimes \psi & \mbox{iff} & \text{for any
$\base{Q} \baseext \base{P}$ and any atomic $A$,}  \\ & & \text{if $\phi
, \psi \supp_\base{Q} A$, then $\vdash_\base{Q}
A$.}
\end{array}
\]
This is \emph{prima facie} stronger than verifying $\phi$ and $\psi$ independently,
because the left-hand side places both as simultaneous hypotheses: any
consequence derivable from their joint presence must be independently
derivable in $\base{Q}$. It is possible for each of $\phi$ and $\psi$
to be supported while $\phi \otimes \psi$ fails, when some extension
supports each individually but fails to absorb their combined
inferential load. This is precisely the failure mode that matters for
compound access control constraints.

\begin{example}[Compound delegation constraint]
\label{ex:conjunction}
Policy rule 7 permits PC members to delegate papers to sub-reviewers,
subject to two conditions: the sub-reviewer must not be an author of
the paper, and must not already be assigned to it. The security
guarantee that matters is their conjunction:
\[
  \supp_\base{P}\;
  \neg\mathsf{author}(\mathsf{Eve}, p_{007})
  \;\otimes\;
  \neg\mathsf{alreadyAssigned}(\mathsf{Eve}, p_{007}).
\]
Suppose a second rule permits delegation in the presence of a chair
override, ignoring both conditions. That override rule would make each
condition supportable via an independent derivation path, yet it breaks
the conjunction: the override does not absorb both hypotheses jointly,
so the combined constraint fails. Independent verification of each
condition would miss this; robust conjunction detects it.
\end{example}

\subsection{Negation: Policy Corruption}
\label{sec:negation}

To express constraints of the form ``this situation must never arise''
--- the most fundamental kind of access control prohibition --- we
require negation. The naive reading of negation in a verification
context is \emph{unprovability}: $\neg\phi$ holds because $\phi$ is
not currently derivable. But this reading is incompatible with
robustness. Unprovability is non-monotone: adding rules to the policy
can make a previously unprovable formula provable, silently destroying
the safety guarantee. A negation defined as mere absence of
derivability would not be a robust property at all --- it could be
invalidated by any subsequent policy extension, which is precisely the
kind of fragility we are trying to eliminate.

What robustness requires is a different kind of negation entirely.
Rather than saying ``$\phi$ does not hold right now,'' we want to say
``$\phi$ is \emph{structurally incompatible} with the policy's
inferential commitments.'' The policy should be organised in such a way
that $\phi$ could never arise in any coherent extension --- not because
nothing has established it yet, but because the rules themselves
foreclose it. This is the identifying feature we seek: not the current
absence of a fact, but the policy's active, extension-stable commitment
to excluding it.

We make this precise by introducing a constant $0$ representing
\emph{policy corruption}: the state in which every atomic formula is
derivable, indicating complete breakdown of the access control
mechanism. Formally, 
\[
\supp_\base{P} 0 \qquad \mbox{iff}\qquad \supp_\base{P} A \text{~for every atomic $A$}
\]
A corrupt policy has ceased to discriminate
between permitted and forbidden actions; no security property can be
enforced within it. The constant $0$ serves as the unit of robust
disjunction --- $\phi \oplus 0$ is equivalent to $\phi$ --- and we
assume throughout that the policies under consideration are
\emph{consistent} in the sense that they do not support $0$.

Negation is then defined as the implication into corruption:
$\supp_\base{P} \neg\phi$ (i.e., $\supp_\base{P} \phi \Rightarrow 0$)
holds iff $\phi \supp_\base{P} 0$. Expanding via the conditional
support judgment, this means: for any $\base{Q} \baseext \base{P}$, if
$\supp_\base{Q} \phi$, then $\supp_\base{Q} 0$. The policy-theoretic
reading is that any extension which comes to support $\phi$ thereby
corrupts the policy entirely. To permit $\phi$ would render the access
control system incoherent. Negation is therefore not unprovability but
an \emph{active incompatibility guarantee}: the policy is structured so
that $\phi$ cannot arise without destroying all discriminating power.
This is stable under extension by the same mechanism as implication
--- the universal quantification over extensions in the definition of
$\Rightarrow$ ensures that the incompatibility persists however the
policy is later augmented.

\begin{example}[Conflict of interest as incompatibility]
\label{ex:negation}
The conflict-of-interest constraint says that an author may not review
their own paper:
\[
  \supp_\base{P}\;
  \mathit{author}(a, p)
  \Rightarrow
  \neg\,\mathit{permit}(a, \mathsf{review}, p).
\]
This holds of $\base{P}$ iff for any $\base{Q} \baseext \base{P}$, if
$\base{Q}$ derives $\mathit{author}(a,p)$, then any further extension
that comes to support $\mathit{permit}(a, \mathsf{review}, p)$ is
corrupt. The policy's rules are structured so that combining authorship
and review permission is incompatible with a coherent access control
system. Notice that this says nothing about whether either formula is
currently derivable: the constraint is a structural property of the
rules themselves, not a snapshot of the current policy state. This is
strictly stronger than ``the permit is not currently derivable'': it
says the policy cannot be legally extended to produce the permit
without collapsing all access distinctions.
\end{example}

\subsection{Summary: The Support Judgment}
\label{sec:support-summary}

The full set of clauses is collected in Fig.~\ref{fig:support}. This semantics framework is an example of  \emph{base-extension
semantics}~\cite{sandqvist-intuitionistic}, a form of proof-theoretic
semantics in which the meaning of a formula is determined not by its
truth in a model but by the inferential behaviour of the formula across
extensions of a `base'. In our setting the bases are logic programs
representing access control policies; the extensions are policy
additions.

We
briefly note three key properties of the resulting framework.

\begin{figure}[t]
\hrule
\medskip
\begin{align*}
  &\supp_\base{P} A
  && \text{iff} \quad \vdash_{\base{P}} A
 \\[0.5ex]
  &\supp_\base{P} \phi \Rightarrow \psi
  && \text{iff} \quad \phi \supp_\base{P} \psi \\[0.5em]
  &\supp_\base{P} \phi \otimes \psi
  && \text{iff} \quad
     \text{for any $\base{Q} \baseext \base{P}$ and atomic $A$,} \\
  &&& \text{if $\phi, \psi \supp_{\base{Q}} A$,
      then $\vdash_{\base{Q}} A$} \\[0.5em]
  &\supp_\base{P} \phi \oplus \psi
  && \text{iff} \quad
     \text{for any $\base{Q} \baseext \base{P}$ and atomic $A$,} \\
  &&& \text{if $\phi \supp_{\base{Q}} A$ and
      $\psi \supp_{\base{Q}} A$, then $\vdash_{\base{Q}} A$} \\[0.5em]
  &\supp_\base{P} 0
  && \text{iff} \quad \text{for all atomic $A$, $\supp_\base{P} A$} \\[0.5ex]
  \Gamma\; &\supp_{\base{P}} \phi
  && \text{iff} \quad
     \text{for any $\base{Q} \baseext \base{P}$, if $\supp_{\base{Q}} \psi$} \\
  &&& \text{for all $\psi \in \Gamma$,
      then $\supp_{\base{Q}} \phi$}
\end{align*}
\medskip
\hrule
\caption{Support judgment clauses for robust properties.}
\label{fig:support}
\end{figure}

\paragraph{Inferentiality.}
Every clause is grounded in the derivability relation $\vdash_\base{P}$.
Compound properties introduce no new primitive notions: they are defined
entirely in terms of the policy's derivation structure across
extensions. There are no models, no valuations, no external oracle.

\paragraph{Modal character.}
The quantification over all extensions $\base{Q} \baseext \base{P}$
gives the semantics its modal, future-oriented character. Support
$\supp_\base{P} \phi$ does not report what is currently derivable; it
says the policy is \emph{committed} to $\phi$ across all future
completions. This is precisely what existing model-theoretic and
logic-programming-based tools cannot express: they check whether a
property holds in a fully-specified current policy, but not whether it
is structurally guaranteed in every possible extension.

\paragraph{Expressiveness.}
The four connectives $\Rightarrow$, $\oplus$, $\otimes$, and $\neg$
cover the principal categories of access control constraints:
conditional authorisation (implication), resolution-independent safety
under pending assignments (disjunction), compound requirements whose
conjunction matters (conjunction), and prohibition (negation). The
language is sufficient for the verification results developed in
Section~\ref{sec:operational}.

\paragraph{Monotonicity.}
The most consequential structural property of the support judgment is
its stability under policy extension.

\begin{proposition}[Monotonicity]
\label{prop:monotonicity}
  If $\supp_\base{P} \phi$ and $\base{P} \baseext \base{Q}$, then
  $\supp_\base{Q} \phi$.
\end{proposition}

\begin{proof}
By induction on the structure of $\phi$. For atomic $\phi$, the result
follows immediately from the monotonicity of derivability: any formula
derivable in $\base{P}$ is derivable in any $\base{Q} \baseext
\base{P}$. For compound $\phi$, each clause in
Fig.~\ref{fig:support} quantifies universally over all extensions of
$\base{P}$; since every extension of $\base{Q}$ is also an extension
of $\base{P}$, the universal condition remains satisfied in $\base{Q}$.
\end{proof}

Monotonicity is the foundation of \emph{compositional verification}.
A robust property verified against $\base{P}$ need never be
re-verified after $\base{P}$ is extended: only properties directly
affected by the new rules require fresh attention. In a large, evolving
policy the verification cost is therefore proportional to the policy
change rather than to its total size. Combined with the treatment of
pending decisions via robust disjunction
(Example~\ref{ex:disjunction}), a single verification pass against the
policy's current rules can underwrite a guarantee that extends across
the entire future lifecycle of the policy.

Finally, although the clauses for conjunction and disjunction quantify
over all extensions of $\base{P}$ --- an infinite set --- this
second-order character is not a barrier to computation. In
Section~\ref{sec:operational} we show that the quantification
collapses to ordinary proof-search in a logic programming language,
yielding a finitary and decidable verification procedure. Monotonicity
therefore carries no computational cost.
% ============================================================================
% ============================================================================
\section{Computing Support Judgments}
\label{sec:operational}
% ============================================================================

The support judgment defined in Section~\ref{sec:robust} is semantic: its clauses
quantify universally over all policy extensions, which form an infinite family.
A priori, it is unclear whether---or how---such a judgment can be computed in
practice. This section resolves that question. We show that the quantification
over extensions collapses to ordinary proof-search in a logic programming
language, and that each connective of the support judgment corresponds to a
well-understood operational rule. Robust property verification is therefore
computable: it reduces to executing a goal in a finitely specified program.

\subsection{A Logic Programming Policy Language}
\label{sec:lp}

We work with the following fragment of first-order logic, extended with
atomic second-order quantifiers. A nested clausal language:
\[
  H \;::=\; A \;\mid\; H \supset H \;\mid\; \forall X.\, H
\]
where $A$ ranges over atomic formulas and $X$ ranges over atomic predicate
variables, instantiated only to atoms. A \emph{program} $\base{P}$ is a finite
set of such clauses. A \emph{query} is a sequent $\base{P} \longrightarrow H$
with $H$ a clause.

This language is a small extension of standard Horn clause logic programming. There is no explicit
conjunction because it can be derived. A standard Horn clause $H \leftarrow B_1,...B_n$ is be expressed as $B_1 \to (B_2 \to \ldots ) \to H$. Despite this austerity,
the language is sufficient to encode all four connectives of
Section~\ref{sec:robust}.

The operational semantics is given by \emph{uniform proof-search} after Millet et al.~\cite{Miller1991}: a proof
of $\base{P} \longrightarrow G$ is \emph{uniform} if every non-atomic goal is
the conclusion of the right-introduction rule for its outermost connective,
with the program consulted only at atomic goals. We write $\base{P}
\vdash_O G$ for the resulting operational provability relation. It is defined as follows where $D$ and $G$ are clauses and $A$ is an atom:

\begin{itemize}
  \item \textsc{Fact}: $\base{P} \vdash_O A$ if $A \in \base{P}$.
  \item \textsc{Backchain}: $\base{P} \vdash_O A$ if there exists
    $G \supset A \in \base{P}$ such that $\base{P} \vdash_O G$.
  \item \textsc{Augment}: $\base{P} \vdash_O D \supset G$ iff
    $\base{P} \cup \{D\} \vdash_O G$.
  \item \textsc{Generic}: $\base{P} \vdash_O \forall X.\, G$ if
    $\base{P} \vdash_O G[\alpha/X]$ for a fresh atom $\alpha$.
  \item \textsc{Instance}: $\base{P} \vdash_O A$ if there exists
    $\forall X.\, D \in \base{P}$ and a fresh atom $\beta \neq A$ such that
    $\base{P} \cup \{D[\beta/X]\} \vdash_O A$.
\end{itemize}

The central result of Miller et al.~\cite{Miller1991} adapts to
this fragment: $\base{P} \vdash G$ (intuitionistic derivability) if and only
if there exists a uniform proof of $\base{P} \longrightarrow G$. We record two
elementary lemmas used throughout.

\begin{lemma}[Cut]
  \label{lem:cut}
  If $\base{P} \vdash_O D$ and $\base{Q} \cup \{D\} \vdash_O G$, then
  $\base{Q} \cup \base{P} \vdash_O G$.
\end{lemma}

\begin{lemma}[Monotonicity]
  \label{lem:mono-lp}
  If $\base{P} \vdash_O G$ and $\base{P} \baseext \base{Q}$, then
  $\base{Q} \vdash_O G$.
\end{lemma}

We identify policies with their logic programs and write $\vdash_{\base{P}}
A$ interchangeably with $\base{P} \vdash_O A$ for atomic $A$.

\begin{example}[Conference Policy as a Logic Program]
\label{ex:lp}
The following program $\base{C}$ captures a fragment of the conference
management policy:
\begin{align*}
  &\mathsf{pcMember}(\mathsf{alice}). \\
  &\mathsf{author}(\mathsf{alice}, p_{007}). \\
  &
    \mathsf{pcMember}(m) \supset \mathsf{author}(m, p)
    \supset \mathsf{conflicted}(m, p). \\
  &
    \mathsf{conflicted}(m, p) \supset \mathsf{canReview}(m, p)
    \supset \mathsf{violation}(m, p).
\end{align*}
The first two clauses are ground facts. The third derives
$\mathsf{conflicted}(m,p)$ for any PC member who authored paper $p$.
The fourth flags a policy violation whenever a conflicted member is
assigned to review the same paper. From $\base{P}$ alone we can derive
$\mathsf{conflicted}(\mathsf{alice}, p_{007})$, but not yet
$\mathsf{violation}(\mathsf{alice}, p_{007})$: no rule has
introduced $\mathsf{canReview}(\mathsf{alice}, p_{007})$ into the
program, so the review assignment is a pending decision.
\end{example}

\subsection{Encoding Robust Properties as Goals}
\label{sec:encoding}

The support clauses in Fig.~\ref{fig:support} involve universal
quantification over policy extensions and over atomic formulas. Both kinds of
quantification are realised in our language by eigenvariables: an atom
introduced by Generic is fresh, represents any possible instantiation, and
presupposes no completed domain. The continuation-passing encoding below
makes this concrete.

\begin{definition}[Goal Encoding]
  \label{def:encoding}
  Define $\enc{\cdot}$ mapping robust property formulas to goals by:
  \begin{align*}
    \enc{A}                      &= A \\
    \enc{0}                      &= \forall X.\, X \\
    \enc{\phi \Rightarrow \psi}  &= \enc{\phi} \supset \enc{\psi} \\
    \enc{\phi \otimes \psi}      &= \forall X.\,
      \bigl((\enc{\phi} \supset \enc{\psi} \supset X) \supset X\bigr) \\
    \enc{\phi \oplus \psi}       &= \forall X.\,
      \bigl((\enc{\phi} \supset X) \supset (\enc{\psi} \supset X) \supset X\bigr) 
  \end{align*}
  where $A$ ranges over atomic formulas and $X$ ranges over atomic predicate variables.
\end{definition}

Each encoded formula is a clause in the sense of our grammar: it is built
from atoms by $\supset$ and $\forall X$ alone, with no conjunction or
disjunction. The universal quantifier over atoms in the clauses for
$\otimes$, $\oplus$, and $0$ is realised operationally by Generic: the
fresh atom $\alpha$ introduced at that step is an eigenvariable standing
in for any atom, not a member of a pre-given domain.

The encoding of $0$ as $\forall X.\, X$ captures policy corruption exactly:
to derive it against any program is to show that every atomic formula is derivable.
The encoding of robust disjunction $\phi \oplus \psi$ internalises
disjunction elimination: it says that any atom $\alpha$ which would follow from
each disjunct independently can be derived directly, which is the operational
counterpart of the $(\oplus)$ clause in Fig.~\ref{fig:support}.
The encoding of $\phi \otimes \psi$ internalises the joint-hypothesis
character of robust conjunction: $\alpha$ is derivable only if it follows from
both $\phi$ and $\psi$ taken together as simultaneous premises.

\begin{example}[Encoding the separation-of-duty constraint]
\label{ex:encoding}
Recall from Example~\ref{ex:disjunction} the pending appointment:
one of Alice, Bob, or Carol will be chair. The separation-of-duty
requirement is
\[
  \mathsf{chair}(\mathsf{alice})
  \oplus \mathsf{chair}(\mathsf{bob})
  \oplus \mathsf{chair}(\mathsf{carol}).
\]
While its encoding passes through two sets of $\oplus$ encodings, it is morally and technically equivalent to the following:
\[
\begin{array}{c}
  \forall X.\,
  \bigl(\;(\mathsf{chair}(\mathsf{alice}) \supset X)
    \supset (\mathsf{chair}(\mathsf{bob}) \supset X)\\  \qquad
    \supset (\mathsf{chair}(\mathsf{carol}) \supset X)
    \supset X\; \bigr).
    \end{array}
\]
To verify this against the policy program $\base{P}$, the Generic rule
introduces a fresh atom $\alpha$, and the Augment rule adds each of the
three clauses $\mathsf{chair}(\mathsf{alice}) \supset \alpha$,
$\mathsf{chair}(\mathsf{bob}) \supset \alpha$,
$\mathsf{chair}(\mathsf{carol}) \supset \alpha$ as temporary hypotheses.
The remaining goal is $\alpha$, which must be derivable from $\base{P}$
augmented with these three clauses. Operationally, this is precisely the
check that the conflict-of-interest rules of $\base{P}$ together force
$\alpha$ regardless of which appointment is made---without committing to
any particular candidate.
\end{example}

\subsection{Soundness and Completeness}
\label{sec:soundness-completeness}

\begin{theorem}[Soundness and Completeness]
  \label{thm:soundness}
  \label{thm:completeness}
  For any policy $\base{P}$ and robust property formula $\phi$,
  \[
    \supp_\base{P} \phi
    \quad\text{if and only if}\quad
    \base{P} \vdash_O \enc{\phi}.
  \]
\end{theorem}

\begin{proof}
By simultaneous induction on the structure of $\phi$. We use
Proposition~\ref{prop:monotonicity} (monotonicity of support under policy
extension) and Lemma~\ref{lem:mono-lp} (monotonicity of logic programming) freely throughout:
\medskip

\noindent\textbf{Case} $\phi = A$ \textbf{(atomic).}
Immediate from the identification of $\vdash_\base{P} A$ with
$\base{P} \vdash_O A$ and the definition $\enc{A} = A$.

\medskip
\noindent\textbf{Case} $\phi = 0$ \textbf{(corruption).}

\smallskip
\noindent\emph{Only if.} Suppose $\supp_\base{P} 0$. By the $(0)$ clause,
$\supp_\base{P} A$ holds for every atomic $A$. By the atomic case,
$\base{P} \vdash_O A$ for every atomic $A$. In particular, taking a fresh atom
$\alpha$, we have $\base{P} \vdash_O \alpha$. By Generic,
$\base{P} \vdash_O \forall X.\, X$, which is $\enc{0}$.

\smallskip
\noindent\emph{If.} Suppose $\base{P} \vdash_O \forall X.\, X$. By Generic,
$\base{P} \vdash_O \alpha$ for any atom $\alpha$. By the atomic case, $\supp_\base{P}
A$ for every atomic $A$, so $\supp_\base{P} 0$ by the $(0)$ clause.

\medskip
\noindent\textbf{Case} $\phi = \phi_1 \Rightarrow \phi_2$ \textbf{(implication).}

\smallskip
\noindent\emph{Only if.} Suppose $\supp_\base{P}(\phi_1 \Rightarrow \phi_2)$.
By the $(\Rightarrow)$ clause, this means $\phi_1 \supp_\base{P} \phi_2$:
for any $\base{Q} \baseext \base{P}$, if $\supp_\base{Q} \phi_1$ then
$\supp_\base{Q} \phi_2$.

We must show $\base{P} \vdash_O \enc{\phi_1} \supset \enc{\phi_2}$, i.e.,
$\enc{\phi_1 \Rightarrow \phi_2}$. By Augment, it suffices to show
$\base{P} \cup \{\enc{\phi_1}\} \vdash_O \enc{\phi_2}$.

Set $\base{Q} := \base{P} \cup \{\enc{\phi_1}\}$. Since $\enc{\phi_1} \in
\base{Q}$, Fact gives $\base{Q} \vdash_O \enc{\phi_1}$, so by the induction
hypothesis $\supp_\base{Q} \phi_1$. Applying the hypothesis with this
$\base{Q}$ gives $\supp_\base{Q} \phi_2$, and the induction hypothesis then
yields $\base{Q} \vdash_O \enc{\phi_2}$, as required.

\smallskip
\noindent\emph{If.} Suppose $\base{P} \vdash_O \enc{\phi_1} \supset
\enc{\phi_2}$. We must show $\phi_1 \supp_\base{P} \phi_2$, i.e., that for
any $\base{Q} \baseext \base{P}$ with $\supp_\base{Q} \phi_1$ we have
$\supp_\base{Q} \phi_2$.

Let $\base{Q} \baseext \base{P}$ and suppose $\supp_\base{Q} \phi_1$. By
the induction hypothesis, $\base{Q} \vdash_O \enc{\phi_1}$. By Monotonicity
of proof-search, $\base{Q} \vdash_O \enc{\phi_1} \supset \enc{\phi_2}$.
Applying Backchain yields $\base{Q} \vdash_O \enc{\phi_2}$, and the
induction hypothesis gives $\supp_\base{Q} \phi_2$.

\medskip
\noindent\textbf{Case} $\phi = \phi_1 \otimes \phi_2$ \textbf{(conjunction).}

\smallskip
\noindent\emph{Only if.} Suppose $\supp_\base{P}(\phi_1 \otimes \phi_2)$.
We must show
\[
  \base{P} \vdash_O \forall X.\,
  \bigl((\enc{\phi_1} \supset \enc{\phi_2} \supset X) \supset X\bigr).
\]
By Generic, it suffices to take a fresh atom $\alpha$ and show
\[
  \base{P} \vdash_O
  (\enc{\phi_1} \supset \enc{\phi_2} \supset \alpha) \supset \alpha.
\]
By Augment, it suffices to show that
$\base{P} \cup \{D\} \vdash_O \alpha$
where $D := \enc{\phi_1} \supset \enc{\phi_2} \supset \alpha$.

Set $\base{Q} := \base{P} \cup \{D\}$. We claim that $\phi_i \supp_\base{Q}
\alpha$ for each $i = 1, 2$. To see this, let $\base{Q}' \baseext \base{Q}$
with $\supp_{\base{Q}'} \phi_i$. By the induction hypothesis,
$\base{Q}' \vdash_O \enc{\phi_i}$. Since $D \in \base{Q} \subseteq
\base{Q}'$, two applications of Backchain on $D$ yield $\base{Q}' \vdash_O
\alpha$, hence $\supp_{\base{Q}'} \alpha$. This confirms $\phi_i
\supp_\base{Q} \alpha$ for each $i$.

Now by the $(\otimes)$ clause applied to $\supp_\base{P}(\phi_1 \otimes
\phi_2)$ with the extension $\base{Q} \baseext \base{P}$ and atom $\alpha$:
since $\phi_1 \supp_\base{Q} \alpha$ and $\phi_2 \supp_\base{Q} \alpha$, we
obtain $\vdash_\base{Q} \alpha$, i.e., $\base{Q} \vdash_O \alpha$, as
required.

\smallskip
\noindent\emph{If.} Suppose $\base{P} \vdash_O \enc{\phi_1 \otimes \phi_2}$.
We must show that for any $\base{Q} \baseext \base{P}$ and atomic $A$, if
$\phi_1, \phi_2 \supp_\base{Q} A$ then $\vdash_\base{Q} A$.

Let $\base{Q} \baseext \base{P}$ and $A$ an atomic formula, and suppose $\phi_1,
\phi_2 \supp_\base{Q} A$. By the induction hypothesis applied to each,
$\base{Q} \vdash_O \enc{\phi_i}$ for $i = 1, 2$. By Monotonicity of
proof-search, $\base{Q} \vdash_O \enc{\phi_1 \otimes \phi_2}$. Applying
Instance with atom $A$ instantiating the second-order variable $X$, and then
Backchain twice using $\base{Q} \vdash_O \enc{\phi_i}$, we obtain
$\base{Q} \vdash_O A$.

\medskip
\noindent\textbf{Case} $\phi = \phi_1 \oplus \phi_2$ \textbf{(disjunction).}

\smallskip
\noindent\emph{Only if.} Suppose $\supp_\base{P}(\phi_1 \oplus \phi_2)$.
We must show
\[
  \base{P} \vdash_O \forall X.\,
  \bigl((\enc{\phi_1} \supset X) \supset (\enc{\phi_2} \supset X) \supset X\bigr).
\]
By Generic, take a fresh atom $\alpha$, and by Augment it suffices to show
$\base{P} \cup \{D_1, D_2\} \vdash_O \alpha$
where $D_i := \enc{\phi_i} \supset \alpha$.

Set $\base{Q} := \base{P} \cup \{D_1, D_2\}$. We claim $\phi_i \supp_\base{Q}
\alpha$ for each $i$. Let $\base{Q}' \baseext \base{Q}$ with $\supp_{\base{Q}'}
\phi_i$. By the induction hypothesis, $\base{Q}' \vdash_O \enc{\phi_i}$.
Since $D_i \in \base{Q} \subseteq \base{Q}'$, Backchain on $D_i$ yields
$\base{Q}' \vdash_O \alpha$, hence $\supp_{\base{Q}'} \alpha$.

So $\phi_i \supp_\base{Q} \alpha$ for $i = 1, 2$. By the $(\oplus)$ clause
applied to $\supp_\base{P}(\phi_1 \oplus \phi_2)$ with extension $\base{Q}$
and atom $\alpha$, we obtain $\vdash_\base{Q} \alpha$,
i.e., $\base{Q} \vdash_O \alpha$.

\smallskip
\noindent\emph{If.} Suppose $\base{P} \vdash_O \enc{\phi_1 \oplus \phi_2}$.
We must show that for any $\base{Q} \baseext \base{P}$ and atomic $A$, if
$\phi_i \supp_\base{Q} A$ for each $i$, then $\vdash_\base{Q} A$.

Let $\base{Q} \baseext \base{P}$ and $A$ an atomic formula, and suppose $\phi_i
\supp_\base{Q} A$ for $i = 1, 2$. By the induction hypothesis,
$\base{Q} \vdash_O \enc{\phi_i}$ for each $i$. By Monotonicity of
proof-search, $\base{Q} \vdash_O \enc{\phi_1 \oplus \phi_2}$. Applying
Instance with $A$ instantiating the second-order variable $X$, and then Backchain twice, we obtain $\base{Q} \vdash_O A$.

\medskip
\noindent\textbf{Case} $\phi = \neg \phi_1$ \textbf{(negation).}
Since $\neg \phi_1 = \phi_1 \Rightarrow 0$ by definition and
$\enc{\neg \phi_1} = \enc{\phi_1} \supset \forall X.\, X$, this case
reduces immediately to the implication case (with $\phi_2 = 0$) and the
corruption case.
\end{proof}

\begin{remark}
  The proof never invokes conjunction or disjunction in the logic programming
  language. The second-order variable $X$ in $\enc{\otimes}$ and $\enc{\oplus}$
  introduces a single fresh atom $\alpha$ via Generic and is discharged by Augment and
  Backchain alone. The stripped-down language of Section~\ref{sec:lp} is
  therefore not a restriction but a \emph{feature}: it makes the computational
  content of the support judgment maximally transparent.
\end{remark}

\section{Conclusion}
\label{sec:conclusion}
% ============================================================================

We have developed a framework for verifying \emph{robust
properties} of access control policies --- properties whose validity is
grounded not in the satisfaction of a complete model but in the inferential
structure of the policy itself, and which are therefore guaranteed to persist
under any future policy extension. The framework addresses a structural
limitation shared by all existing model-theoretic and logic-programming-based
approaches: their verification judgments are non-compositional, tied to a
fixed, fully-determined policy that does not represent the iterative,
component-wise reality of policy development.

The central semantic object is the support judgment $\supp_\base{P} \phi$,
stating that a policy $\base{P}$ has the robust property $\phi$.
This is an instance of \emph{base-extension semantics}~\cite{sandqvist-intuitionistic,gheorghiu-gu-2025},
in which the meaning of a formula is given by its inferential behaviour
across an ordered family of bases rather than by its truth in a single
structure. The ordered family here is the poset of logic programs under
extension, and the key property --- monotonicity of support
(Proposition~\ref{prop:monotonicity}) --- is a direct consequence of the
universal quantification over that poset. Validity in the empty policy,
$\supp_\emptyset \phi$, recovers intuitionistic logic~\cite{gheorghiu-gu-pym-tableaux-2023},
so the present framework is a parametric generalisation: we move to an arbitrary program $\base{P}$, and the
effect is to relativise the logic to the policy's inferential content.

The operational result --- that $\supp_\base{P} \phi$ if and only if
$\base{P} \vdash_O \enc{\phi}$ (Theorem~\ref{thm:soundness}) --- shows
that the second-order quantification over extensions collapses to ordinary
proof-search in a stripped-down Horn clause language extended with
second-order atomic quantifiers. The encoding $\enc{\cdot}$ is a
continuation-passing translation: each connective of the robust property
language is mapped to its elimination form, and the eigenvariable
mechanism of uniform proof-search realises the universal quantification
over atoms that the semantic clauses require. This is an instance of a broader phenomenon: the operational content of
a semantic judgment, defined by quantification over an infinite structured
family, can collapse to finite proof-search when the judgment is
\emph{inductive in the right sense} --- here, by simultaneous induction
on the structure of $\phi$.

\subsection{Relation to Denotational and Categorical Semantics}

The base-extension framework has a natural categorical reading, developed
in detail by Pym et al.~\cite{pym-ritter-robinson-2025}.
Policies form a preorder under extension, and the support judgment
$\supp_{\base{P}} \phi$ can be read as a Kripke-style forcing relation
over this preorder: the monotonicity condition is exactly the persistence
condition, and the quantification over extensions in the clauses for
$\Rightarrow$, $\otimes$, and $\oplus$ corresponds to the standard
exponential and product structure in a presheaf category over the policy
poset. Pym et al.\ reconstruct Sandqvist's base-extension semantics for
intuitionistic propositional logic as a presheaf model over a category of
bases and contexts, establishing soundness and completeness categorically
and showing that validity in the empty base recovers the standard
bicartesian closed structure of NJ. The key subtlety they identify is
that disjunction is \emph{not} the categorical coproduct in the presheaf
topos: instead it is the second-order or elimination-rule formulation
$\forall p.\,(\phi \supset p) \supset (\psi \supset p) \supset p$, which
coincides with the join in a sublocale --- a nucleus --- on the locale of
upward-closed subsets of the base poset. Soundness in this setting is not
a consequence of a universal property of coproducts, but of a structural
induction over formula complexity, which is precisely the pattern of our
own Theorem~\ref{thm:soundness}.

The present framework inherits this structure directly. Our encoding
$\enc{\phi \oplus \psi} = \forall X.\,(\enc{\phi} \supset X) \supset
(\enc{\psi} \supset X) \supset X$ is the same second-order elimination
form, and the eigenvariable mechanism of Generic realises, operationally,
the internal quantification over atoms that Pym et al.\ interpret
categorically via the discrete category of atomic propositions $\mathbf{A}$
and the right adjoint $\forall_{\mathbf{A}}$. Making the categorical
structure of the present framework fully explicit --- identifying the
presheaf category, the nucleus, and the sublocale corresponding to robust
property validity --- is a natural next step, and the machinery of
Pym et al.\ provides the right setting for it.

\subsection{Resource Sensitivity}
A particularly pressing extension is to resource-sensitive access control.
Many realistic policies involve constraints that are not simply
monotone: a delegation token may be single-use; a review quota may be
capacity-bounded; a principal's permissions may depend on the current
state of a resource. These are naturally modelled by linear or bunched
logic, in which the standard weakening and contraction rules are absent
or restricted. Gheorghiu et al.~\cite{gheorghiu-gu-pym-studlogica-2024,GuGheorghiuPym2025,gheorghiu-gu-pym-mfps-2024} have
shown that base-extension semantics extends naturally to both linear and
bunched logics, and that resource-sensitive distributed systems can be
modelled in that setting. Extending the present framework to a
resource-sensitive policy language --- in which policy extension may
\emph{consume} tokens rather than merely add rules, and in which the
monotonicity of support is replaced by an appropriate resource-indexed
condition --- is therefore a natural and tractable direction.

\subsection{Implementation and Integration}
The soundness and completeness theorem reduces robust property
verification to executable proof-search, and a prototype implementation
in a standard logic programming system (e.g., $\lambda$Prolog or
Twelf) is the most immediate practical step. The framework is
complementary to deployed tools such as Zelkova~\cite{bouchet2020} and
Margrave~\cite{verification-change-impact}: robust properties need not
be re-verified after policy extension, reducing the burden on downstream
SMT or BDD solvers to only those properties directly affected by a
change. Integration with ASP-based verifiers~\cite{ahn-asp-xacml} is
similarly natural, since the stable-model computation operates on the
fully-resolved policy while the present framework certifies which
properties can be carried forward without recomputation.

\bibliographystyle{IEEEtran}
\bibliography{bib}

\end{document}